\documentclass[oneside,english]{amsart}
\usepackage[T1]{fontenc}
\usepackage[utf8]{inputenc}
\usepackage{amsthm}
\usepackage{amssymb}

\makeatletter
\numberwithin{equation}{section}
\numberwithin{figure}{section}
\theoremstyle{plain}
\newtheorem{thm}{\protect\theoremname}
  \theoremstyle{plain}
  \newtheorem{lem}[thm]{\protect\lemmaname}
\newenvironment{lyxlist}[1]
{\begin{list}{}
{\settowidth{\labelwidth}{#1}
 \setlength{\leftmargin}{\labelwidth}
 \addtolength{\leftmargin}{\labelsep}
 }}
{\end{list}}
  \theoremstyle{plain}
  \newtheorem{prop}[thm]{\protect\propositionname}
  \theoremstyle{plain}
  \newtheorem{cor}[thm]{\protect\corollaryname}
  \theoremstyle{plain}
  \newtheorem*{thm*}{\protect\theoremname}

\usepackage{algorithmic}

\makeatother

\usepackage{babel}
  \providecommand{\corollaryname}{Corollary}
  \providecommand{\lemmaname}{Lemma}
  \providecommand{\propositionname}{Proposition}
  \providecommand{\theoremname}{Theorem}
\providecommand{\theoremname}{Theorem}

\begin{document}

\title{Run Vector Analysis and Barker Sequences of Odd~Length}

\author{Jürgen Willms}

\email{willms.juergen@fh-swf.de}

\address{Institut für Computer Science, Vision and Computational Intelligence,
Fachhochschule Südwestfalen, D-59872 Meschede, Germany}
\begin{abstract}
The run vector of a binary sequence reflects the run structure of
the sequence, which is given by the set of all substrings of the run
length encoding. The run vector and the aperiodic autocorrelations
of a binary sequence are strongly related. In this paper, we analyze
the run vector of skew-symmetric binary sequences. Using the derived
results we present a new and different proof that there exists no
Barker sequence of odd length n >13. Barker sequences are binary sequences
whose off-peak aperiodic autocorrelations are all in magnitude at
most 1.
\end{abstract}

\keywords{binary sequence, Barker sequence, autocorrelation, run vector, run
structure, run, run length encoding, balanced, skew-symmetric}

\subjclass[2000]{11B83, 94A55, 68R15, 68P30 }

\date{30.10.14}

\maketitle

\section{Introduction}

Let $a=(a_{1},a_{2},\cdots,a_{n})$ be a sequence of real numbers
of length $n\geq1$. $a$ is called a \emph{binary} sequence if $a_{i}\in\{-1,1\}$
for all $i=1,\cdots,n$. For $k=0,1,\cdots,n-1$ the $k$th \emph{aperiodic
autocorrelation} of the binary sequence$a$ is defined by

\begin{equation}
C_{k}:=\sum_{i=1}^{n-k}a_{i}a_{i+k}.\label{eq:ck}
\end{equation}
Furthermore, we set $C_{n}:=0$. $C_{0}$ is called the \emph{peak
autocorrelation} of $a$ and we have $C_{0}=n$ and 
\begin{equation}
C_{k}\equiv n-k\bmod2\quad\mbox{for}\: k=1,2,\cdots,n.\label{eq:Ckmod2}
\end{equation}

In a wide range of engineering applications it is of interest to collectively
minimize the absolute values of the off-peak autocorrelations. Prominent
examples of such binary sequences are Barker sequences whose off-peak
aperiodic autocorrelations are in magnitude as small as possible.
Thus, a Barker sequence is a binary sequence for which $|C_{k}|\leq1$
for $k=1,2,\cdots,n-1$. If $n$ is odd, then a well-known result
in \cite{turyn1961binary} says, that Barker sequences of odd length
$n>13$ do not exist. However, in \cite{willms2014counterexamples}
it is shown that the proof of this result as presented in \cite{turyn1961binary}
is incomplete. It is at the moment not all at clear whether and how
this rather complicated proof can be fixed. 

Similar to \cite{borwein2013note}, one objective of this paper is
to find an alternative proof that is easier to understand. A main
idea of the original proof in \cite{turyn1961binary} is to show that
the first elements of a Barker sequence of odd length exhibit some
periodic behavior, which at the same time implies that the sequence
must be short. Run length encoding is a natural way to express this
“periodic behavior”. Since autocorrelation and run structure of a
binary sequence are strongly related (see \cite{willms2013RunStructure6506978}),
another objective of this paper is to explore whether the properties
of the run vector lead to an alternative proof. 

The here presented alternative proof that there exists no Barker sequence
of odd length $n>13$, relies on a careful analysis of the run vector
of skew-symmetric sequences. The run vector of a binary sequence (as
introduced in \cite{willms2013RunStructure6506978}) reflects the
run structure which is given by the set of all substrings of the run
length encoding. In order to simplify the definition of the run vector
we use a slightly different approach compared to \cite{willms2013RunStructure6506978}
where a more combinatorial approach is used. As shown in \cite{willms2013RunStructure6506978}
skew-symmetric sequences have a balanced run length encoding and vice
versa. Theorem 1 (see \cite{willms2013RunStructure6506978}) shows
how the run vector and the aperiodic correlations are related. For
the reader's convenience these two used results from \cite{willms2013RunStructure6506978}
are explicitly proven in Appendix A by using the here introduced terminology. 

Lemma \ref{lem:balanced_RkOdd} characterizes sequences having a balanced
run length in terms of their run vector; in particular, this result
is used in Proposition \ref{prop:Barker-Balanced} that a Barker sequence
of odd length has a balanced run length encoding. A key result is
Lemma \ref{lem:RkplusRk}, which can be used in order to show that
a Barker sequence of odd length must be short. Theorem \ref{thm:AllBarker}
describes the structure of the first elements of a Barker sequence
of odd length $n$ provided that $n>5$ and $r_{1}>1$; it also gives
an upper bound for $n$. The following Corollary \ref{cor:No-Barker}
states that Barker sequences of odd length $n>13$ do not exist. Finally,
assuming that $a_{1}=a_{2}$ we show in Corollary \ref{thm:barkerRLE-1-1}
that for a Barker sequence of odd length $n\geq3$ either $r=(2,1)$,
$r=(3,1,1)$, $r=(3,2,1,1)$, $r=(3,3,1,2,1,1)$ or $r=(5,2,2,1,1,1,1)$
where $r$ denotes the run length encoding of $a$. 

As usual, we exploit the fact that periodic and aperiodic autocorrelations
are related. For the definition of the periodic autocorrelations put
$a_{n+i}:=a_{i}$ for $i\geq1$; then for $k=0,1,\cdots,n-1$ the
$k$th \emph{periodic autocorrelation }is given by

\begin{equation}
\tilde{C}_{k}:=\sum_{i=1}^{n}a_{i}a_{i+k}.\label{eq:pck}
\end{equation}
Note further that for $k=1,\cdots,n-1$ we have $\tilde{C}_{k}=\tilde{C}_{n-k}$
and 
\begin{equation}
\tilde{C}_{k}=C_{k}+C_{n-k}.\label{eq:Ckperiodic}
\end{equation}

If the binary sequence $a$ is constant then $\tilde{C}_{k}=n$. Moreover,
inverting one element $a_{i}$ of $a$ either changes the sum $\tilde{C}_{k}$
by $\pm4$ or leaves $\tilde{C}_{k}$ unchanged. Therefore, for $k=0,\cdots,n-1$
we obtain

\begin{equation}
\tilde{C}_{k}\equiv n\bmod4.\label{eq:Ckmod4}
\end{equation}

Note that if $a$ is a Barker sequence of odd length $n=2m-1$ then
by (\ref{eq:Ckperiodic}) and (\ref{eq:Ckmod4}) we obtain for $k=1,2,\cdots,n-1$ 

\begin{equation}
C_{k}=\begin{cases}
0\; & \mbox{if }k\mbox{ odd}\\
(-1)^{m+1}\; & \mbox{\mbox{if }\ensuremath{k}\mbox{ even}.}
\end{cases}\label{eq:CkBarker}
\end{equation}

\pagebreak{}

\section{Preliminaries }

In the following $a$ will always be a fixed binary sequence of length
$n$. To circumvent boundary problems we will put $a_{0}:=0$. Furthermore,
$r=(r_{1},r_{2},\cdots,r_{\gamma})$ will always denote the run length
encoding of $a$ and $\gamma$ the length of $r$. In particular,
we have $\sum_{j=1}^{\gamma}r_{j}=n=C_{0}$. Moreover, if we put $s_{0}:=0$
and $s_{j+1}:=s_{j}+r_{j+1}$ then $a_{s_{j}}\neq a_{s_{j}+1}$ and
$a_{s_{j}+1}=a_{s_{j}+2}=\cdots=a_{s_{j+1}}$ for all $j=0,1,\cdots,\gamma-1$.
Note that $1\leq s_{1}<s_{2}<\cdots<s_{\gamma}=n$ and 
\begin{equation}
s_{j}=r_{1}+r_{2}+\cdots+r_{j}\quad\mbox{for }j=1,2,\cdots,\gamma\label{eq:s-def}
\end{equation}

Observe that

\begin{equation}
C_{1}=1+n-2\gamma\label{eq:C1_gamma}
\end{equation}
since for each $j=1,2,\cdots,\gamma-1$ the $j$th run of $a$ contributes
$r_{j}-2$ to the sum $C_{1}(a)$ and the last run contributes $r_{\gamma}-1$
to the sum $C_{1}$; thus $C_{1}=1+\sum_{j=1}^{\gamma}(r_{j}-2)=1+n-2\gamma$.
In \cite{willms2013RunStructure6506978} it is shown that the run
structure of $a$ can be used in order to efficiently calculate the
autocorrelations of $a$. In order to formulate these results we need
a few more definitions from \cite{willms2013RunStructure6506978}.

Let $t_{0}:=0$ and $t_{j+1}:=t_{j}+r_{\gamma-j}$ for $j=0,1,\cdots,\gamma-1$.
Note that $1\leq t_{1}<t_{2}<\cdots<t_{\gamma}=n$ and that for $j=1,2,\cdots,\gamma$
\begin{equation}
t_{j}:=r_{\gamma}+r_{\gamma-1}+\cdots+r_{\gamma-j+1}\label{eq:t-def}
\end{equation}

and

\begin{equation}
s_{j}+t_{\gamma-j}=n.\label{eq:s+t}
\end{equation}
As in \cite{willms2013RunStructure6506978} let 
\begin{equation}
S:=\{s_{1},s_{2},\cdots,s_{\gamma-1}\}\label{eq:S_def}
\end{equation}
\begin{equation}
T:=\{t_{1},t_{2},\cdots,t_{\gamma-1}\}.\label{eq:T_def}
\end{equation}

Note that for $1\leq k\leq n-1$ we have 
\begin{equation}
k\in S\Leftrightarrow a_{k}a_{k+1}=-1\quad\mbox{and}\ensuremath{\quad k\in S\Leftrightarrow n-k\in T.}\label{eq:Sakplus1}
\end{equation}

In the following consider the functions $f_{S},\: f_{T},\: f:\;\mathbb{Z}\rightarrow\{-1,0,1\}$
defined by
\begin{equation}
f_{S}(k):=\begin{cases}
(-1)^{j}\; & \mbox{if }k\in S\mbox{ with \ensuremath{k=s_{j}}}\\
0\; & \mbox{otherwise}
\end{cases}\label{eq:f-def}
\end{equation}
\begin{equation}
f_{T}(k):=\begin{cases}
(-1)^{j}\; & \mbox{if }k\in T\mbox{ with \ensuremath{k=t_{j}}}\\
0\; & \mbox{otherwise}
\end{cases}\label{eq:g-def}
\end{equation}

\begin{equation}
f(k):=f_{S}(k)+f_{T}(k).\label{eq:def-f}
\end{equation}

Note that by (\ref{eq:s+t}) we have for all $k\mathbb{\in Z}$ 
\begin{equation}
f_{S}(k)=(-1)^{\gamma}f_{T}(n-k).\label{eq:f=00003Dg}
\end{equation}
\pagebreak{}

For $k=1,2,\cdots,n-1$ define

\begin{equation}
U_{k}:=\sum_{j=1}^{\gamma-1}(-1)^{j}f_{T}(k-s_{j}).\label{eq:def-Uk}
\end{equation}

Note that $U_{k}=0$ if $1\leq k\leq s_{1}.$ Next we will define
the run vector $R=(R_{1},R_{2},\cdots,R_{n-1})$. The run vector $R$
of a binary sequence reflects the run structure, which is given by
the set of all substrings of $r$. Unlike in \cite{willms2013RunStructure6506978},
where a more combinatorial approach based on the run structure is
chosen, we will use $U_{k}$ in order to define $R_{k}$; therefore,
for $k=1,2,\cdots,n-1$ define 

\begin{equation}
\tilde{R}_{k}:=f(k)+2U_{k}\qquad\mbox{and}\qquad R_{k}:=(-1)^{\gamma}\tilde{R}_{n-k}.\label{eq:rk-def}
\end{equation}

The following result from \cite{willms2013RunStructure6506978} which
we will use for the proof of our main result shows how the aperiodic
autocorrelations and the run vector are related.
\begin{thm}[see \cite{willms2013RunStructure6506978}]
\label{thm:main_aperiodic-1}$\mbox{Let }k=1,\cdots,n-1$; then
\end{thm}
\[
C_{k+1}-2C_{k}+C_{k-1}=-2R_{k}
\]

\begin{proof}
Corollary 8 in \cite{willms2013RunStructure6506978} shows that $R_{k}$
is $k$th component of the\emph{ }run vector\emph{ }$R$ of $a$ and
thus the result follows directly from Theorem 1 in \cite{willms2013RunStructure6506978}. 

In Appendix A we give a separate proof of Theorem 1 without referring
to \cite{willms2013RunStructure6506978}. 
\end{proof}

\subsection{Skew-Symmetric Sequences and Balanced Run Length Encoding}

An odd length binary sequence $a$ of length $2m-1$ is called \emph{skew-symmetric}
if for all $j=1,2,\cdots,m-1$ 
\begin{equation}
a_{m-j}=(-1)^{j}a_{m+j}.\label{eq:skew}
\end{equation}

Skew-symmetric sequences are of particular interest in different areas;
as we will see Barker sequences of odd length are for example are
skew-symmetric. It is not difficult to see that for a skew-symmetric
sequence $a$ we have $C_{k}=0$ if $k$ is odd. In particular, by
(\ref{eq:C1_gamma}) it follows that $\gamma=\frac{n+1}{2}=m$. Furthermore,
by Theorem \ref{thm:main_aperiodic-1} we have $C_{k+1}-2C_{k}+C_{k-1}=-2R_{k}$
for all $k=1,2,\cdots,n-1.$ Therefore for a skew-symmetric sequence
$a$ and $1\leq k<n$ we obtain $R_{k}=C_{k}$ if $k$ is even and
$R_{k}=-\frac{1}{2}(C_{k-1}+C_{k+1})$ if $k$ is odd.

As in \cite{willms2013RunStructure6506978} we call a run length encoding
$r$ \emph{balanced} if
\begin{equation}
S\cup T=\{1,2,\cdots,n-1\}\;\mbox{and }\; S\cap T=\varnothing.\label{eq:SvT}
\end{equation}
 Note that if $r$ is balanced then $2(\gamma-1)=n-1$ and hence $n$
is odd and $\gamma=\frac{n+1}{2}$. Note that this definition of a
balanced run length encoding differs considerably from the definition
of ``balanced binary sequences'' in the terminology of \cite{golomb1967shift},
where it is required that the number of ones is nearly equal to the
number of minus ones. In \cite{willms2013RunStructure6506978} it
is shown that a binary sequence $a$ is skew-symmetric if and only
if its run length encoding $r$ is balanced; a separate proof of this
result is given in Proposition \ref{prop:-balanced-is-skew-symmetric}
in Appendix A.\pagebreak{}
\begin{lem}
\label{lem:balanced_RkOdd} $r$ is balanced if and only if $\tilde{R}_{k}$
is odd for all \textup{$1\leq k<n$ }\end{lem}
\begin{proof}
If $r$ is balanced then for all $k=1,2,\cdots,n-1$ either $k\in S$
or $k\in T$; thus $|f(k)|$= 1 and by (\ref{eq:rk-def}) $\tilde{R}_{k}$
is odd. On the other hand, suppose that $\tilde{R}_{k}$ is odd for
all $1\leq k<n$. Then $f(k)$ is odd and thus $|f(k)|=1$ for all
$1\leq k<n$; hence either $k\in S\setminus T$ or $k\in T\setminus S$
. Therefore, $S\cap T=\varnothing$ and $S\cup T=\{1,2,\cdots,n-1\}$,
which shows that $r$ is balanced. \end{proof}
\begin{lem}
\label{lem:fT}Let $r$ be balanced and assume that $s_{\mu-1}<k<s_{\mu}$
for some $1\leq\mu\leq\gamma$. Then $f(k)=f_{T}(k)=-(-1)^{k+\mu}$. \end{lem}
\begin{proof}
Let $s_{\mu-1}<k<s_{\mu}$ for some $1\leq\mu\leq\gamma.$ Since $r$
is balanced, we have $k\in T$ and thus $k=t_{j}$ for some $1\leq j<\gamma$.
Since $r$ is balanced we have $j=k-(\mu-1)$ and hence $f(k)=f_{T}(k)=(-1)^{j}=-(-1)^{k+\mu}$.\end{proof}
\begin{lem}
\label{lem:FunctionG}Let $r$ be balanced and let $s\in S$ be odd;
then
\begin{lyxlist}{00.00.0000}
\item [{(i)}] $f(s-1)=-f(s)\quad$ if $s>1$ 
\item [{(ii)}] $f(s+1)=f(s)\quad$ if $s+1\in T$. 
\end{lyxlist}
\end{lem}
\begin{proof}
We have $s=s_{\mu}$ for some $1\leq\mu<\gamma$. Note that then $f(s)=(-1)^{\mu}.$
If $s-1\in S$ then $f(s-1)=-f(s)$ by the definition of $f_{S}$.
If $s-1\in T$ then by Lemma \ref{lem:fT} we have $f(s-1)=-(-1)^{s-1+\mu}=-f(s)$
since $s$ is odd. Similarly, if $s+1\in T$ then we get by Lemma
\ref{lem:fT} that $f(s+1)=-(-1)^{s+1+\mu+1}=f(s)$.
\end{proof}
Note, that a similar result holds if $s\in S$ is even, but in following
we are only interested in the case where $s$ is odd.
\begin{lem}
\label{lem:Ukmod}Let $r$ be balanced and let $1\leq k<n$. Assume
that $s_{\mu-1}<k\leq s_{\mu}$ for some $0\leq\mu\leq\gamma$. Then

\[
U_{k}\equiv\mu\pmod2\quad\Leftrightarrow\quad k=2s\;\mbox{for some }\ensuremath{s\in S}.
\]
 \end{lem}
\begin{proof}
Put $M_{k}$$:=\{1\leq j<\gamma:\:|f_{T}(k-s_{j})|=1\}$ and $m_{k}:=|M_{k}|$.
Then $U_{k}\equiv m_{k}\pmod2$ and $m_{k}<\mu$; furthermore, $m_{k}=\mu-1$
if and only if $k-s_{j}\in T$ (and thus $k-s_{j}\notin S$) for all
$j=1,2,\cdots,\mu-1.$ Observe that if $k-s_{j}\in S$ with $j<\mu$
then $k-s_{j}=s_{i}$ for some $1\leq i<\mu$ and hence also $k-s_{i}\in S$.
Moreover, $k-s_{j}=s_{j}$ if and only if $k=2s$ for some $s\in S$.
Therefore, we have $m_{k}\equiv\mu\pmod2$ (and thus $U_{k}\equiv m_{k}\pmod2$)
if and only if $k=2s$ for some $s\in S$. \end{proof}
\begin{lem}
\label{lem:2kinSmod4}Let $r$ be balanced and let $1\leq k<n$ with
$k\in S$. Then $k=2s$ for some $s\in S$ if and only if $\tilde{R}_{k}\equiv1\pmod4$. \end{lem}
\begin{proof}
We have $k=s_{\mu}$ for some $1\leq\mu<\gamma$ and thus $f(k)=f_{S}(k)=(-1)^{\mu}$
and $\tilde{R}_{k}=(-1)^{\mu}+2U_{k}$. Since $(-1)^{\mu}\equiv2\mu+1\pmod4$
we have $\tilde{R}_{k}\equiv1\pmod4$ if and only if $U_{k}\equiv\mu\pmod2$.
Hence, by Lemma \ref{lem:Ukmod} it follows that $k=2s$ for some
$s\in S$ if and only if $\tilde{R}_{k}\equiv1\pmod4$.
\end{proof}
The next lemma shows that if $r$ is balanced and the first $\mu-1$
elements of $r$ are greater than $1$ then the last $s_{\mu}-\mu$
elements are less than or equal to $2$.
\begin{lem}
\label{lem:balTail}Let $r$ be balanced and let $1\leq\mu<\gamma$.
Assume that $r_{j}\geq2$ for all $j=1,2,\cdots,\mu-1$. Then $s_{\mu}-\mu<\gamma$
and $r_{\gamma+1-j}\leq2$ for all $j=1,2,\cdots,s_{\mu}-\mu$.\end{lem}
\begin{proof}
Since $r$ is balanced we have $|\{1,2,\cdots,s_{\mu}\}\cap S|=\mu$
and $|\{1,2,\cdots,s_{\mu}\}\cap T|=s_{\mu}-\mu$. In particular,
this shows that $s_{\mu}-\mu<\gamma$ and $t_{s_{\mu}-\mu}<s_{\mu}$.
Assume that $r_{\gamma+1-j}>2$ for some $1\leq j\leq s_{\mu}-\mu$.
Then $t_{j}-t_{j-1}=r_{\gamma+1-j}>2$. Since $r$ is balanced $t_{j}-1$
and $t_{j}-2$ are both in $S$; hence $t_{j}-1=s_{i}$ for some $2\leq i\leq\gamma$
and thus $s_{i}-1=s_{i-1}$. Furthermore, since $j\leq s_{\mu}-\mu$
we have $s_{i}=t_{j}-1\leq t_{s_{\mu}-\mu}-1<s_{\mu}$, which shows
that $i<\mu$ and thus $r_{i}\geq2$. On the other hand, since $s_{i}-1=s_{i-1}$
we have $r_{i}=1$, which contradicts our assumption. This shows that
$r_{\gamma+1-j}\leq2$ for all $j=1,2,\cdots,s_{\mu}-\mu$.
\end{proof}

\subsection{Barker Sequences}

Let $a$ be a Barker sequence of odd length $n=2m-1$ and let $1\leq$$k<n.$
Then by (\ref{eq:CkBarker}) and Theorem \ref{thm:main_aperiodic-1}
we obtain
\begin{eqnarray}
R_{1} & = & \begin{cases}
-m & \mbox{if }m\mbox{ odd}\\
1-m & \mbox{if }m\mbox{ even}
\end{cases}\label{eq:RkBarkerNeu}\\
R_{k} & = & (-1)^{k+m+1}\quad\mbox{for }k=2,3,\cdots,n-1.\nonumber 
\end{eqnarray}
The next result provides the basis for the analysis of the run vector
of Barker sequences of odd length. 
\begin{prop}
\label{prop:Barker-Balanced}Let $a$ be a Barker sequence of odd
length $n$. Then $r$ is balanced, $\gamma=\frac{n+1}{2}$ and 
\begin{eqnarray}
\tilde{R}_{k} & = & (-1)^{k}\quad\mbox{for all}\mbox{ }1\leq k\leq n-2\label{eq:Rk=00003D1^k}\\
\tilde{R}_{n-1} & = & \begin{cases}
\gamma & \mbox{if }\gamma\mbox{ odd}\\
1-\gamma & \mbox{if }\gamma\mbox{ even}.
\end{cases}\nonumber 
\end{eqnarray}
\end{prop}
\begin{proof}
Let $n=2m-1$ and $1\leq k<n.$ By (\ref{eq:RkBarkerNeu}) $R_{k}$
and thus also $\tilde{R}_{k}$ is odd; hence by Lemma \ref{lem:balanced_RkOdd}
$r$ is balanced . Since $r$ is balanced, we have $\gamma=\frac{n+1}{2}$;
hence $\gamma=m$ and (\ref{eq:Rk=00003D1^k}) follows directly from
(\ref{eq:RkBarkerNeu}). 
\end{proof}
Note that by Proposition \ref{prop:-balanced-is-skew-symmetric} a
Barker sequence of odd length is skew-symmetric. Note further that
this result can be generalized. If $n$ is odd and $C_{j}=0$ for
all odd $j$ with $1\leq j<n$ then from (\ref{eq:Ckperiodic}), (\ref{eq:Ckmod4})
and Theorem \ref{thm:main_aperiodic-1} it follows that $\tilde{R}_{k}$
is odd for all $1\leq k<n$ and thus $r$ is balanced and $a$ skew-symmetric.

\specialsection{The Main Result}

In this section we will prove that there exists no Barker sequence
of odd length $n>13$. For the following we always set $p:=r_{1}$,
and if $r$ is balanced with $p>1$ we always set
\[
\nu:=min\{1\leq j<\gamma:\; r_{j+1}\bmod p\neq0\}\quad\mbox{and \ensuremath{\quad q:=r_{\nu+1}}}.
\]
Let $r$ be balanced with $p>1$. Then $r_{\gamma}=1$; hence $\nu$
is well defined and $1\leq\nu<n$. In the following we will consider
two different cases, which we can treat quite similarly by putting
$\alpha:=1$ if $q\bmod p=1$ and $\alpha:=0$ otherwise. 

Let us first sketch the idea of the proof. Suppose that $a$ is a
Barker sequence of odd length $n$. As we will see we may assume without
loss of generality that $a_{1}=a_{2}$; thus we have $p>1.$ The key
observation is that on the one hand by (\ref{eq:Rk=00003D1^k}) we
have $\tilde{R}_{k}+\tilde{R}_{k-1}=0$ for all$\mbox{ }k=2,3,\cdots,n-2$.
On the other hand, we show in Lemma \ref{lem:RkplusRk} that $|\tilde{R}_{k_{0}}+\tilde{R}_{k_{0}-1}|\geq2$
if $k_{0}:=p+s_{\nu+1}+\alpha<n$. Hence, we must have $k_{0}\geq n-1$,
which is, as we will see, only possible if $p=3$ or $p=5$ and if
in addition $\nu\leq2$ and $q\leq2$. However, this implies that
$n$ is small. 
\begin{lem}
\label{lem:sj_mod_q}Let $r$ be balanced with $p>1$; then
\begin{lyxlist}{00.00.0000}
\item [{(i)}] $\tilde{R}_{p}=-1$ 
\item [{(ii)}] \textup{$s_{j}\bmod p=0$ }for all \textup{$j=1,2,\cdots,\nu$}
\item [{(iii)}] \textup{$s_{\nu+1}\bmod p\neq0$ }and \textup{$q\bmod p\neq0$ }
\item [{(iv)}] \textup{$s_{\nu+1}\leq\gamma+1$ }if $p\geq3$
\end{lyxlist}
\end{lem}
\begin{proof}
\emph{(i)}: From the definition of $U_{p}$ and $f$ it follows directly
that $U_{p}=0$ and $f(p)=f_{S}(p)=-1$; hence, $\tilde{R}_{p}=-1$. 

\emph{(ii)} and \emph{(iii)}: Both statements follow directly from
the definition of $\nu$.

\emph{(iv)}: Let $p\geq3$ and put $\mu:=\nu+1;$ then $\mu<\gamma$
and $r_{j}\geq p\geq3$ for all $j=1,2,\cdots,\mu-1$ and $\mu<\gamma$.
By Lemma \ref{lem:balTail} we have $s_{\mu}-\mu<\gamma$ and $r_{\gamma+1-(s_{\mu}-\mu)}\leq2$.
Hence, $\gamma+1-(s_{\mu}-\mu)\geq\mu$ and thus $s_{\mu}\leq\gamma+1$.
\end{proof}
The next lemma is a key result; it provides the main argument for
showing that a Barker sequence of odd length must be short.
\begin{lem}
\label{lem:RkplusRk}Let $r$ be balanced with $p\geq3$ and let $k:=p+s_{\nu+1}+\alpha$.
Assume that
\begin{lyxlist}{00.00.0000}
\item [{(i)}] $p$ and \textup{$s_{\nu+1}$} are both odd
\item [{(ii)}] \textup{$s_{\nu+1}+1\in T\quad$ }if \textup{$\alpha=1$.}
\end{lyxlist}
Then
\begin{equation}
|\tilde{R}_{k}+\tilde{R}_{k-1}|\geq2\quad\mbox{if }k<n\label{eq:RkRkMinus1}
\end{equation}
 \end{lem}
\begin{proof}
Put $k_{j}:=k-s_{j}$ for $j=1,2,\cdots,\gamma-1$ and $g(i):=f_{T}(i)+f_{T}(i-1)$
for $i\mathbb{\in Z}$. Then by (\ref{eq:rk-def}) we have $\tilde{R}_{k}+\tilde{R}_{k-1}=f(k)+f(k-1)+2\sum_{j=1}^{\gamma-1}(-1)^{j}g(k_{j})$.

We claim that $g(k_{1})=(-1)^{\nu+\alpha}$ and $g(k_{\nu+1})=(-1)^{\alpha}.$
Note that $k_{1}=s_{\nu+1}+\alpha$ and $k_{\nu+1}=p+\alpha$. If
$\alpha=0$ then $k_{1}\in S$, $k_{1}$ is odd by (\emph{i}) and
$k_{1}-1\in T$ since $q>1$; therefore, by Lemma \ref{lem:FunctionG}
we obtain $g(k_{1})=f_{T}(k_{1})+f_{T}(k_{1}-1)=f_{T}(k_{1}-1)=f(k_{1}-1)=-f(k_{1})=(-1)^{\nu}$
and $g(k_{\nu+1})=f_{T}(p)+f_{T}(p-1)=1$ since $p$ is odd. 

If $\alpha\neq0$ then by (\emph{ii}) $k_{1}\in T$ and $k_{1}-1\in S$
and by Lemma \ref{lem:FunctionG} we get $g(k_{1})=f_{T}(k_{1})=f(s_{\nu+1}+1)=f(s_{\nu+1})=(-1)^{\nu+1}$.
Observe that $p+1\in T$ since $p+1\in S$ would imply that $s_{\nu+1}=s_{2}$
is even. Hence, $g(k_{\nu+1})=f_{T}(p+1)+f_{T}(p)=-1$ for $\alpha\neq0$. 

Next we want to show that $g(k_{j})=0$ for $2\leq j\leq\nu$. So
suppose that $\nu\geq2$ and let $2\leq j\leq\nu$. Then $r_{2}\geq p$
and thus $k_{j}<s_{\nu+1}$. Furthermore, note that $k_{j}=p+(s_{\nu}-s_{j})+q+\alpha$
and that $p$  divides neither $q+\alpha$ nor $q+\alpha-1.$ Therefore,
$p$ divides neither $k_{j}$ nor $k_{j}-1$; since $k_{j}<s_{\nu+1}$this
shows that $k_{j}$ and $k_{j}-1$ are both in $T$ and thus $g(k_{j})=0$. 

Finally, suppose that $j\geq\nu+2$. Then by (\emph{ii}) we have $k_{j}<p$.
Note that $g(1)=-1$ and $g(i)=0$ if $i<p$ and $i\neq1$. Moreover,
we have $k_{j}=1$ if and only if $f_{S}(k-1)\neq0$, and if $k_{j}=1$
then $g(k_{j})=-1$ and $f_{S}(k-1)=f(s_{j})=(-1)^{j}$. Therefore,
we get $\sum_{j=\nu+2}^{\gamma-1}(-1)^{j}g(k_{j})=-f_{S}(k-1)$. 

Altogether, we get $\sum_{j=1}^{\gamma-1}(-1)^{j}g(k_{j})=-(-1)^{\nu+\alpha}+(-1)^{\nu+1}(-1)^{\alpha}-f_{S}(k-1)=-2(-1)^{\nu+\alpha}-f_{S}(k-1)$
and thus $\tilde{R}_{k}+\tilde{R}_{k-1}=f(k)+f_{T}(k-1)-f_{S}(k-1)-4(-1)^{\nu+\alpha}$.
Since $|f(k)|=1$ and $|f_{T}(k-1)-f_{S}(k-1)|=1$ we get $|\tilde{R}_{k}+\tilde{R}_{k-1}|\geq2$.
\end{proof}
Note that if $r$ is balanced then $2|\tilde{R}_{k}+\tilde{R}_{k-1}|=|C_{n-k+2}-C_{n-k}|$
if $k$ is odd and $2|\tilde{R}_{k}+\tilde{R}_{k-1}|=|C_{n-k+1}-C_{n-k-1}|$
if $k$ is even. Thus, $|\tilde{R}_{k}+\tilde{R}_{k-1}|\geq2$ implies
that $|C_{j}|\geq3$ for some $j\in\{n-k-1,n-k,n-k+1,n-k+2\}$. 

Next we want to apply the derived results to Barker sequences of odd
length. Note that if $a$ is a Barker sequence of odd length $n$
then Proposition \ref{prop:Barker-Balanced} shows that $r$ is balanced
and $\tilde{R}_{k}=(-1)^{k}$ for all $1\leq k\leq n-2$. 
\begin{lem}
\label{lem:p_odd}Let $a$ be a Barker sequence of odd length $n$
with $p>1$; then
\begin{lyxlist}{00.00.0000}
\item [{(i)}] $p$ is odd (and thus $p\geq3)$ if \textup{$n>3$}
\item [{(ii)}] $s_{\nu+1}$ is odd if \textup{$n>5$ }
\item [{(iii)}] \textup{$s_{\nu+1}+1\in T$ }if \textup{$n>5$ }and \textup{$\alpha=1$}
\item [{(iv)}] \textup{$n\leq p+s_{\nu+1}+\alpha+1$ if }$n>5$.
\end{lyxlist}
\end{lem}
\begin{proof}
(\emph{i}): Suppose $n>3$; then $\gamma\geq3$ and thus $p<n-1$
and $\tilde{R}_{p}=(-1)^{p}$. Hence, by Lemma \ref{lem:sj_mod_q}
(\emph{i}) $p$ is odd and hence $p\geq3$.

(\emph{ii}): Suppose that $n>5$. Then by Lemma \ref{lem:sj_mod_q}
(\emph{iv}) we have $s_{\nu+1}\leq\gamma+1\leq\frac{n+3}{2}\leq n-2$.
Assume that $s_{\nu+1}$ is even. Then $\tilde{R}_{s_{\nu+1}}=1$
and by Lemma \ref{lem:2kinSmod4} we would get $s_{\nu+1}=2s_{j}$
for some $1\leq j\leq\nu$, which is not possible by Lemma \ref{lem:sj_mod_q}
(\emph{ii}) and (\emph{iii}). Hence, $s_{\nu+1}$ is odd. 

(\emph{iii}): Suppose that $n>5$ and $\alpha=1$. By (\emph{ii})
$s_{\nu+1}+1$ is even. Assume that $s_{\nu+1}+1\in S$; then by Lemma
\ref{lem:2kinSmod4} we have $s_{\nu+1}+1=2s_{j}$ for some $1\leq j\leq\nu$
and thus $q\bmod p=s_{\nu+1}\bmod p=p-1$, which is not possible since
$p\geq3$ and $q\bmod p=1$. Hence, $s_{\nu+1}+1\in T$. 

(\emph{iv}): Note that $\tilde{R}_{k}+\tilde{R}_{k-1}=0$ for all$\mbox{ }2\leq k\leq n-2$.
Suppose that $n>5$. Then (\emph{i}), (\emph{ii}) and (\emph{iii})
show that we can apply Lemma \ref{lem:RkplusRk} which gives us that
$n-1\leq p+s_{\nu+1}+\alpha$. \end{proof}
\begin{thm}
\label{thm:AllBarker}Let $a$ be a Barker sequence of odd length
$n>5$ with $p>1$; then $n\leq13$ and either $r_{1}=r_{2}=3$ and
$r_{3}=1$ or $r_{1}\in\{3,5\}$ and $r_{2}=2$.\end{thm}
\begin{proof}
Lemma \ref{lem:sj_mod_q} (\emph{iv}) together with Lemma \ref{lem:p_odd}
(\emph{i}) and (\emph{iv}) give us that $2s_{\nu+1}\leq2\gamma+2=n+3\leq p+s_{\nu+1}+\alpha+4$;
thus we have $\nu p+q\leq s_{\nu+1}\leq p+\alpha+4$ and therefore
$\nu\leq\frac{p+4-(q-\alpha)}{p}$. Hence, $1\leq\nu\leq2$ and $q\leq4$.
Moreover, this shows also that $\nu=2$ is only possible if $p=3$,
$r_{2}=3$ and $q=1$. Hence for $\nu=2$ we have $r_{1}=r_{2}=3$,
$r_{3}=1$, $\alpha=1$ and $s_{3}=7$ and therefore $n\leq12$ by
Lemma \ref{lem:p_odd} (\emph{iv}).

Next suppose that $\nu=1$. By Lemma \ref{lem:p_odd} (\emph{ii})
$s_{2}$ is odd and so $q$ must be even; hence either $q=2$ or $q=4$.
In particular, we have $n\leq2p+q+1$ by Lemma \ref{lem:p_odd} (\emph{iv}).We
claim that $p\leq5.$ Assume that $p>5$. Then $p\geq7$ and $r_{\gamma}=r_{\gamma-1}=\cdots=r_{\gamma-5}=1$
and therefore $\{n-5,n-4,\cdots,n-1\}\subseteq S$. Lemma \ref{lem:2kinSmod4}
gives us that $2s=n-5$ for some $s\in S$. Again by Lemma \ref{lem:2kinSmod4}
it follows that $s+1$ is also in $S$. Since $q\geq2$ we have $s\geq p+q$
and thus $n\geq2p+2q+5$, which is not possible. This shows that $p\leq5$.

It remains to show that if $\nu=1$ then $q=4$ is not possible. If
$p=5$, $\nu=1$ and $q=4$, then $r=(5,4,\cdots,1,1,2,1,1,1,1)$
because $r$ is balanced; but this is not possible since then $n\leq2p+q+1=15$
by Lemma \ref{lem:p_odd} (\emph{iv}). Similarly, if $p=3$, $\nu=1$
and $q=4$ then $r=(3,4,\cdots,1,1,2,1,1)$, which is not possible
since then $n\leq2p+q+1=11$. Hence, $q=2$ and $n\leq13$ if $\nu=1$.\pagebreak{}\end{proof}
\begin{cor}
\label{cor:No-Barker}There exists no Barker sequences of odd length
$n>13$.\end{cor}
\begin{proof}
Let $a$ be a Barker sequence of odd length $n$. Note that if $n>$1
and $p=r_{1}=1$ then $r_{\gamma}>1$. Note further that the binary
sequence $(a_{n},a_{n-1},\cdots,a_{1})$, whose run length encoding
is given by $(r_{\gamma},r_{\gamma-1},\cdots,r_{1})$ is also a Barker
sequence of the same length $n$. Hence, without loss of generality
we can assume that $p>1$. By Theorem \ref{thm:AllBarker} we have
$n\leq13$. \end{proof}
\begin{cor}
\label{thm:barkerRLE-1-1} $a$ is a Barker sequence of odd length
$n$ with $p>1$ if and only if $r=(2,1)$ , $r=(3,1,1)$, $r=(3,2,1,1)$,
$r=(3,3,1,2,1,1)$ or $r=(5,2,2,1,1,1,1)$. \end{cor}
\begin{proof}
It is easy to check that if $r=(2,1)$ , $r=(3,1,1)$, $r=(3,2,1,1)$,
$r=(3,3,1,2,1,1)$ or $r=(5,2,2,1,1,1,1)$ then $a$ is a Barker sequence
of odd length $n$ with $p>1$. 

Conversely, let $a$ be a Barker sequence of odd length $n$ with
$p>1$. Proposition \ref{prop:Barker-Balanced} shows that $r$ is
balanced. Thus, $n=3$ implies $r=(2,1)$ and $n=5$ implies $r=(3,1,1)$
since by Lemma \ref{lem:p_odd} (\emph{i}) $p$ is odd if $n>3$.
For $n>5$ Theorem \ref{thm:AllBarker} shows that either $r_{1}=r_{2}=3$
and $r_{3}=1$ or $r_{1}\in\{3,5\}$ and $r_{2}=2$. We will consider
these three cases separately. 

Suppose that $r_{1}=r_{2}=3$ and $r_{3}=1$. Then $r=(3,3,1,\cdots,2,1,1)$
and thus $n\geq11$. On the other hand, Lemma \ref{lem:p_odd} (\emph{iv})
shows that $n\leq12$ . Hence, we have $n=11$ and $r=(3,3,1,2,1,1)$.

Suppose next that $r_{1}=3$ and $r_{2}=2$. Then $r=(3,2,\cdots,1,1)$
and $r_{\gamma-2}=2$. Moreover, $n\leq9$ by Lemma \ref{lem:p_odd}
(\emph{iv}) and thus $\gamma\leq5$. Therefore, either $r=(3,2,1,1)$
or $r=(3,2,2,1,1)$. If $r=(3,2,2,1,1)$ were a Barker sequence then
we would have by (\ref{eq:Rk=00003D1^k}) $\tilde{R}_{6}=1$ but (\ref{eq:rk-def})
gives us $\tilde{R}_{6}=f(6)-2(f_{T}(3)+f_{T}(1))=(-1)^{6+2}-2=-1.$
Hence, $(3,2,2,1,1)$ is not a Barker sequence. Therefore, we have
in this case $r=(3,2,1,1)$.

Finally, suppose that $r_{1}=5$ and $r_{2}=2$. Then $r=(5,2,\cdots,1,1,1,1)$
and $r_{\gamma-4}=2$. Moreover, $n\leq13$ by Lemma \ref{lem:p_odd}
(\emph{iv}) and thus $\gamma\leq7$. Therefore, either $r=(5,2,1,1,1,1)$
or $r=(5,2,2,1,1,1,1)$. However, if $r=(5,2,1,1,1,1)$ then $s_{3}=8$
is even and by (\ref{eq:Rk=00003D1^k}) $\tilde{R}_{8}=1$, which
contradicts Lemma \ref{lem:2kinSmod4}. Hence, $(5,2,1,1,1,1)$ is
not a Barker sequence. Therefore, we have in this case $r=(5,2,2,1,1,1,1)$. 
\end{proof}

\specialsection{Appendix A}

For the reader's convenience we prove in this appendix all used results
from \cite{willms2013RunStructure6506978}. 
\begin{thm*}
Let $1\leq k<n$; then
\[
C_{k+1}-2C_{k}+C_{k-1}=-2R_{k}
\]
\end{thm*}
\begin{proof}
Put $a_{0}:=0$, $a_{n+1}:=0$ and $\delta_{i}:=a_{i}-a_{i-1}$ for
all $i=1,2,\cdots,n+1$. Furthermore define $C_{k}(\delta):=\sum_{i=1}^{n+1-k}\delta_{i}\delta_{i+k}$.
Let $1\leq k<n$; then as in the proof of Lemma 2 in \cite{willms2013RunStructure6506978}

\begin{eqnarray*}
C_{k}(\delta) & = & \sum_{i=1}^{n+1-k}(a_{i}-a_{i-1})(a_{i+k}-a_{i+k-1})\\
 & = & \sum_{i=1}^{n+1-k}a_{i}a_{i+k}-a_{i}a_{i+k-1}-a_{i-1}a_{i+k}+a_{i-1}a_{i+k-1}\\
 & = & (C_{k}(a)+a_{n+1-k}a_{n+1})-C_{k-1}(a)\\
 &  & -(a_{0}a_{k+1}+C_{k+1}(a)+a_{n-k}a_{n+1})+(a_{0}a_{k}+C_{k}(a))\\
 & = & -C_{k+1}(a)+2C_{k}(a)-C_{k-1}(a).
\end{eqnarray*}

Note that $\delta_{1}=a_{1}$, $\delta_{n+1}=(-1)^{\gamma}a_{1}$
and for $i=1,2,\cdots,n-1$ we have $2f_{S}(i)=a_{1}\delta_{i+1}$.
Hence, as in the proof of Theorem 6 in \cite{willms2013RunStructure6506978}
we have on the other hand

\begin{eqnarray*}
\frac{1}{2}C_{k}(\delta) & = & \frac{1}{2}(\delta_{1}\delta_{k+1}+\delta_{n+1-k}\delta_{n+1}+\sum_{i=2}^{n-k}\delta_{i}\delta_{i+k})\\
 & = & f_{S}(k)+(-1)^{\gamma}f_{S}(n-k)+2\sum_{i=1}^{n-k-1}f_{S}(i)f_{S}(i+k)\\
 & = & f_{S}(k)+f_{T}(k)+2\sum_{j=1}^{\gamma-1}f_{S}(s_{j})\cdot f_{S}(s_{j}+k)\\
 & = & f(k)+2\sum_{j=1}^{\gamma-1}(-1)^{j}f_{S}(s_{j}+k)\\
 & = & (-1)^{\gamma}f(n-k)+2\sum_{j=1}^{\gamma-1}(-1)^{j}(-1)^{\gamma}f_{T}(n-(k+s_{j})\\
 & = & (-1)^{\gamma}\tilde{R}_{n-k}=R_{k}
\end{eqnarray*}

which proves the theorem.\end{proof}
\begin{prop}
\label{prop:-balanced-is-skew-symmetric}$a$ is skew-symmetric if
and only if $r$ is balanced.\end{prop}
\begin{proof}
Let $1\leq k<n$. Suppose that $a$ is skew-symmetric. Then $a_{k}a_{k+1}=-a_{n+1-k}a_{n-k}$
and thus by (\ref{eq:Sakplus1}) we have $k\in S\Leftrightarrow a_{k}a_{k+1}=-1\Leftrightarrow a_{n+1-k}a_{n-k}=1\Leftrightarrow n-k\notin S\Leftrightarrow k\notin T.$
Thus, $S\cap T=\varnothing$ and also $S\cup T=\{1,2,\cdots,n-1\}$.
Hence, $r$ is balanced.

Suppose that run length encoding $r$ of $a$ is balanced. Then $n-k\in T\Leftrightarrow n-k\notin S$
and thus by (\ref{eq:Sakplus1}) $a_{k}a_{k+1}=-1\Leftrightarrow k\in S\Leftrightarrow n-k\in T\Leftrightarrow n-k\notin S\Leftrightarrow a_{n-k}a_{n-k+1}=1$.
Hence, $a_{k}a_{k+1}=-a_{n-k}a_{n-k+1}$. Since $n=2\gamma-1$ we
therefore have $a_{\gamma-i}a_{\gamma-i+1}=-a_{\gamma+i-1}a_{\gamma+i}$
for all $i=1,2,\cdots,\gamma-1$. In particular, for $i=1$ we get
$a_{\gamma-1}=-a_{\gamma+1}$ and thus $a_{\gamma-2}=a_{\gamma+2}$
for $i=2$, $a_{\gamma-3}=-a_{\gamma+3}$ for $i=3$ and so on. Hence,
$a$ is skew-symmetric. 
\end{proof}
\bibliographystyle{ieeetr}
\bibliography{rleBarker1}

\bigskip{}

\end{document}